\documentclass[11pt]{article}

\usepackage{amssymb,amsmath,amsfonts}
\usepackage{amsthm}
\usepackage{times,inconsolata}
\usepackage{fullpage}
\usepackage{latexsym}
\usepackage{verbatim}
\usepackage{enumerate}
\usepackage{graphicx}
\usepackage{url}
\usepackage[ruled,linesnumbered,noend]{algorithm2e}
\usepackage{color}

\newcommand{\Ch}{\Children}
\newcommand{\Children}{\ensuremath{{\mathrm{Ch}}}}

\newcommand{\Cl}{\ensuremath{{\mathrm{Cl}}}} 		

\newcommand{\MP}{\ensuremath{M_\P}} 
\newcommand{\TG}{\ensuremath{\Gamma}} 
\newcommand{\incompatible}{\texttt{incompatible}}
\newcommand{\Uinit}{\ensuremath{U_\mathrm{init}}} 
\newcommand{\Ubef}{\ensuremath{U_\mathrm{bef}}} 
\newcommand{\Uaft}{\ensuremath{U_\mathrm{aft}}} 
\newcommand{\Urem}{\ensuremath{U_\mathrm{rem}}} 
\newcommand{\Yinit}{\ensuremath{Y_\mathrm{init}}} 

\newcommand{\Build}{\ensuremath{\textsc{Build}}} 
\newcommand{\BuildG}{\ensuremath{\textsc{BuildST}}} 

\newcommand{\Hcal}{\ensuremath{\mathcal{H}}}
\renewcommand{\P}{\ensuremath{\mathcal{P}}}
\newcommand{\R}{\ensuremath{\mathcal{R}}}

\newcommand{\COUNT}{\ensuremath{\mathtt{count}}} 
\newcommand{\LIST}{\ensuremath{\mathtt{List}}} 
\newcommand{\SEMI}{\ensuremath{\mathtt{singleton}}} 

\newtheorem{theorem}{Theorem}
\newtheorem{lemma}{Lemma}

\newtheorem{observation}{Observation}

\begin{document}
\title{Fast Compatibility Testing for Rooted Phylogenetic Trees\thanks{Supported in part by the National Science Foundation under grants CCF-1017189 and CCF-1422134.}}
\author{
Yun Deng\thanks{Department of Computer Science, Iowa State University, Ames, IA 50011, USA, {\tt yundeng@iastate.edu}}
\and
David Fern\'{a}ndez-Baca\thanks{Department of Computer Science, Iowa State University, Ames, IA 50011, USA, {\tt fernande@iastate.edu}}
}

\date{\empty}

\maketitle

\begin{abstract}
We consider the following basic problem in phylogenetic tree construction. Let $\P = \{T_1, \ldots, T_k\}$ be a collection of rooted phylogenetic trees over various subsets of a set of species.    The tree compatibility problem asks whether there is a tree $T$ with the following property: for each $i \in \{1, \dots, k\}$, $T_i$ can be obtained from the restriction of $T$ to the species set of $T_i$ by contracting zero or more edges.  If such a tree $T$ exists, we say that $\P$ is compatible.  

We give a $\tilde{O}(\MP)$ algorithm for the tree compatibility problem, where $\MP$ is the total number of nodes and edges in $\P$.  Unlike previous algorithms for this problem, the running time of our method does not depend on the degrees of the nodes in the input trees.  Thus, it is equally fast on highly resolved and highly unresolved trees.\end{abstract}

\section{Introduction}

Building a phylogenetic tree that encompasses all living species is one of the central challenges of computational biology.  Two obstacles to achieving this goal are lack of data and conflict among the data that is available.  The data shortage is tied to the vast disparity in the amount of information at our disposal for different families of species and the limited amount of comparable data across families \cite{Sanderson:2008}.  One approach to overcoming this obstacle begins by identifying subsets of species for which enough data is available, and building phylogenies for each subset.  The resulting trees are then synthesized into a single phylogeny ---a supertree--- for the combined set of species.  This approach, proposed
in the early 90s \cite{Baum:1992,Ragan:1992}, has been used successfully to build large-scale phylogenies (see, e.g., \cite{BinindaEmonds:Nature:07,HinchliffPNAS2015}).  

Any attempt at synthesizing phylogenetic information from multiple input trees must deal with the potential for conflict among these trees. Conflict may arise due to errors, or due to phenomena such as gene duplication and loss, and horizontal gene transfer.  A fundamental question is whether conflict exists at all; that is, does there exist a supertree that exhibits the evolutionary relationships implicit in each input tree?  We can formalize this question as follows.  Let $\P = \{T_1, \ldots, T_k\}$ be a collection  of rooted phylogenetic trees, where, for each $i \in \{1, \dots , k\}$, $T_i$ is a phylogenetic tree for a set of species $L(T_i)$.  The \emph{tree compatibility problem} asks whether there exists a phylogenetic supertree $T$ for the set of species $\bigcup_{i = 1}^k L(T_i)$ such that, for each $i \in \{1, \dots , k\}$, $T_i$ can be obtained from $T|L(T_i)$ --- the minimal subtree of $T$ spanning $L(T_i)$ --- by zero or more contractions of internal edges (that is, $T|L(T_i)$ is homeomorphic to $T$).  If the answer is ``yes'', then $\P$ is said to be \emph{compatible}; otherwise, $\P$ is \emph{incompatible}.  

Here we present an algorithm that solves the compatibility problem for rooted trees in $O(\MP \log^2 \MP)$ time, where $\MP$ is the total number of vertices and edges in the trees in $\P$.   This running time is independent of the degrees of the internal nodes of the input trees.


\paragraph{Previous work.}  Aho et al.\ \cite{Aho81a} gave the first polynomial-time algorithm for the rooted tree compatibility problem.  Their motivation was not phylogenetics, but relational data\-bases.  Steel \cite{Steel92} was perhaps the first to note the relevance of Aho et al.'s algorithm to supertree construction. His version of the Aho et al.\ algorithm,  which he called the \Build\ algorithm, has been a major influence in later work, including the present paper. 

Henzinger et al.\ \cite{HenzingerKingWarnow99} showed that one can check the compatibility of a collection $\R$ of rooted triples --- that is, phylogenetic trees on three species --- in $O(|\R| \log^2 |\R|)$ time.  (The time bound stated in \cite{HenzingerKingWarnow99} is higher, but can be improved using a faster dynamic graph connectivity data structure \cite{HolmLichtenbergThorup:2001}.)  Any collection of trees $\P$ can be encoded by a collection of rooted triples $\R(\P)$, obtained by enumerating the restriction of each input tree to every three-element subset of its species set (see Section \ref{sec:prelims}).  If $n$ denotes the total number of distinct species in $\P$, then we get a trivial upper bound of $|\R(P)| = O(n^3 k)$.  We can improve on this by finding a \emph{minimal} set $\R^*$ of rooted triples that define the input trees.  If the trees are binary --- \emph{fully resolved}, in the language of phylogenetics ---, then $O(n)$ triples suffice for each tree, giving us $|\R^*| = O(n k)$.  If input trees admit non-binary --- that is, \emph{unresolved} --- nodes, however, the number of triples needed per input tree is roughly proportional to $n^2$ (the precise bound depends on the sum of the products of the degrees of internal nodes and the degrees of their children \cite{GrunewaldSS:2007}), giving us $|\R^*| = O(n^2 k)$.  Of course, the extra step of finding $\R^*$ adds to the complexity of the algorithm.

The tree compatibility problem is related to the \emph{incomplete directed perfect phylogeny problem} (IDPP).  Indeed, any collection of $k$ phylogenetic trees on $n$ distinct species can be encoded as a problem of testing the compatibility of a collection of $O(\MP)$ ``directed partial characters'' on $n$ species\footnote{For a precise definition of partial characters and IDPP, we refer the reader to Pe'er et al.\ \cite{PeerShamirSharan04}.}. Intuitively, each such character encodes the species in the subtree rooted at some node in an input tree.  There is a $\tilde{O}(nm)$ algorithm to test the compatibility of $m$ incomplete characters\ \cite{PeerShamirSharan04}, which can be adapted to yield a $\tilde{O}(n\MP)$ algorithm for tree compatibility.  

When the input trees are unrooted, the tree compatibility problem becomes NP-hard \cite{Steel92}.  Nevertheless, the decision version is polynomial-time solvable if $k$ is fixed \cite{BryantLagergren06}; that is, the problem is fixed-parameter tractable in $k$.  The proof of fixed-parameter tractability in \cite{BryantLagergren06} relies on Courcelle's Theorem \cite{Courcelle90}, and thus is an existence proof, rather than a practical algorithm.  

\paragraph{Our contributions.}
At a high level, our algorithm resembles \Build\ \cite{Steel92,SempleSteel03}.  There are, however, important differences.  \Build\ relies on the \emph{triplet graph}, whose nodes are the species and where there is an edge between two species if they are involved in a triplet (see Section \ref{sec:prelims}).  Our algorithm relies instead on  intersection graphs of sets of species associated with certain nodes of the input trees.  
Our graphs allow a more compact representation of the triplets induced by the trees in $\P$ (see Section \ref{sec:testCompat}).  The key to the correctness of our approach is the intimate relationship between the triplet graph and our intersection graph  (see Lemma \ref{lm:Wi} of Section \ref{sec:testCompat}).  We remark that intersection graphs have a long history of use in testing compatibility, beginning with the work of Buneman \cite{Bun74}.

We also take ideas from other sources.  From Pe'er et al.'s IDPP algorithm \cite{PeerShamirSharan04}, we adapt the idea of a \emph{semi-universal node}.  Although the graphs used to solve IDPP and rooted compatibility are different, semi-universal nodes play similar roles in each case: they capture the notion of sets of nodes in the input trees that map to the same node in a supertree, if a supertree exists.   The relationship between our algorithm and Pe'er et al.'s goes deeper.  Our approach can be viewed as an algorithm for IDPP that takes advantage of the fact that our particular set of incomplete characters arises from a collection of trees.

Intersection graphs are a convenient tool to prove the correctness for our algorithm.  They are less convenient for an implementation, because it is hard to maintain them dynamically, in the way our algorithm requires. The difficulty lies in recomputing set intersections whenever the graphs are updated.  We avoid this by using \emph{display graphs}, an idea that we borrow from the proof of the fixed-parameter tractability of unrooted compatibility  \cite{BryantLagergren06}.  The display graph of a collection $\P$ is obtained by identifying leaves in the input trees that have the same label.  Display graphs provide all the connectivity information we need for our intersection graphs (see Lemma~\ref{lm:GnH} of Section~\ref{sec:implementation}), but are easier to maintain.  

Through our techniques, we achieve what, to our knowledge, is the first algorithm for rooted compatibility  to achieve near-linear time under all input conditions, regardless of the degrees of the nodes in the input trees.  This is an essential quality for dealing with large datasets.

\paragraph{Contents.}
Section \ref{sec:prelims} reviews basic concepts in phylogenetics, defines compatibility formally, and introduces triplets and the triplet graph.
Section \ref{sec:testCompat} presents our intersection graph approach to testing tree compatibility. Section \ref{sec:implementation} describes the implementation details needed to achieve the $O(\MP \log^2 \MP)$ time bound.  Section \ref{sec:discussion} contains some final remarks.

\section{Preliminaries}
\label{sec:prelims}

For each positive integer $r$, $[r]$ denotes the set $\{1, \dots , r\}$.

\paragraph{Phylogenetic trees.} 
Let $T$ be a rooted tree.   We use $V(T)$, $E(T)$, and $r(T)$ to denote the nodes, edges, and the root of $T$, respectively. For each $x \in V(T)$, we use $\Ch(x)$ and $T(x)$ to denote the set of children of $x$ and the subtree of $T$ rooted at $x$, respectively.  Suppose $u, v \in V(T)$.  Then, $u$ is a \emph{descendant} of $v$ if $v$ lies on the path from $u$ to $r(T)$ in $T$.  Note that $v$ is a descendant of itself.  $T$ is \emph{binary}, or \emph{fully resolved}, if each of its internal nodes has two children.

A (rooted) \emph{phylogenetic tree} is a rooted tree $T$ where every internal node has at least two children, along with a bijection $\lambda$ that maps each leaf of $T$ to an element of a set of \emph{species}, denoted by $L(T)$. 
For each $x \in V(T)$, $L(x)$ denotes the set of species mapped to the leaves of $T(x)$; that is, $L(x) = \{\lambda(v) : v \text{ is a leaf in } T(x)\}$. $L(x)$ is called the \emph{cluster} at $x$.  Note that $L(r(T)) = L(T)$. 
The set of all clusters in $T$ is $\Cl(T) = \{L(x) : x \in V(T)\}$.   

The following lemma, adapted from \cite[p.~52]{SempleSteel03}, is part of the folklore of phylogenetics.

\begin{lemma}\label{lm:compat}
Let $\Hcal$ be a collection of non-empty subsets of a set of species $X$ that includes all singleton subsets of $X$ as well as $X$ itself.  If there exists a phylogenetic tree $T$ such that $\Cl(T) = \Hcal$, then, up to isomorphism, $T$ is unique.
\end{lemma}


Let $T$ be a phylogenetic tree and $A$ be a set of species.  The \emph{restriction} of $T$ to $A$, denoted $T{|A}$ is the phylogenetic tree with species set $A$ where $\Cl(T|A) = \{C \cap A : C \in \Cl(T) \text{ and } C \cap A \neq \emptyset\}$.  Let $T'$ be a phylogenetic tree.  $T$ \emph{displays} $T'$ if $\Cl(T') \subseteq \Cl(T|L(T'))$.  Equivalently, $T$ displays $T'$ if $T|L(T')$ is homeomorphic to $T'$.

A \emph{rooted triple} is a binary phylogenetic tree on three leaves. A rooted triple with leaves $a$, $b$, and $c$ is denoted $ab|c$ if the path from $a$ to $b$ does not intersect the path from $c$ to the root.  We treat $ab|c$ and $ba|c$ as equivalent. 

When restricted to the three-element subsets of its species set, a phylogenetic tree $T$ induces a set $\R(T)$ of rooted triples, defined as $\R(T) = \{T|X : X \subseteq L(T),  |X| = 3 \text{ and } T|X \text{ is binary}\}$. 

\begin{lemma}[{\cite[p.~119]{SempleSteel03}}]\label{lm:triples}
Let $T$ and $T'$ be two phylogenetic trees.  Then $T$ displays $T'$ if and only if $\R(T') \subseteq \R(T)$.
\end{lemma}

\paragraph{Profiles and compatibility.}
Throughout the rest of this paper $\P = \{T_1, \dots, T_k\}$ denotes a set where, for each $i \in [k]$, $T_i$ is a phylogenetic tree.  We refer to $\P$ as a \emph{profile}, and write $L(\P)$ to denote $\bigcup_{i\in[k]} L(T_i)$, the \emph{species set} of $\P$.  
We write $V(\P)$ for $\bigcup_{i\in[k]} V(T_i)$, $E(\P)$ for $\bigcup_{i\in[k]} E(T_i)$, 
and $\R(\P)$ for $\bigcup_{i \in [k]} \R(T_i)$. 
Given a subset $A$ of $L(P)$, $\P|A$ denotes the profile $\{T_1|A, \dots, T_k|A\}$.
The \emph{size} of $\P$ is $\MP = |V(\P)| + |E(\P)|$.  Note that $\MP = O(nk)$.

Profile $\P$ is \emph{compatible} if there exists a phylogenetic $X$-tree $T$ such that, for each $i \in [k]$, $T$ displays $T_i$.  If such a tree $T$ exists, we say that $T$ \emph{displays} $\P$.  See Figure~\ref{fig:Profile}.
\begin{figure}
\begin{center}
\includegraphics[scale=0.34]{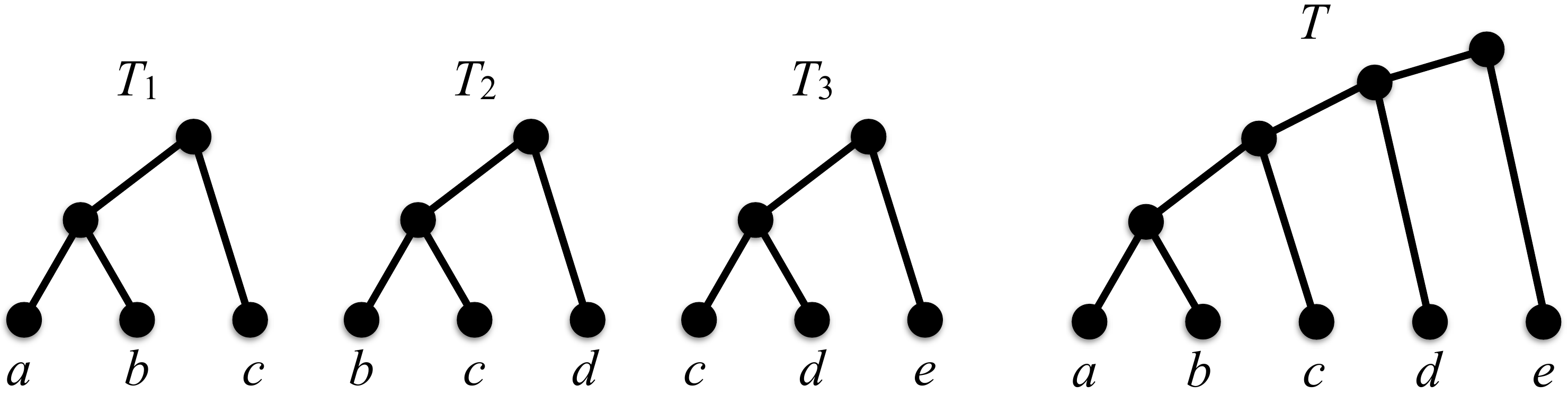}
\end{center}
\vspace{-4\parsep}
\caption{A profile $\P = \{T_1, T_2, T_3\}$ and a tree $T$ that displays $\P$.}
\label{fig:Profile}
\end{figure}

\paragraph{The triplet graph.}   The \emph{triplet graph} of a profile $\P$, denoted $\TG(\P)$, is the 
graph whose vertex set is $L(P)$ 
and where there is an edge between species $a$ and $b$ if and only if there
exists a $c \in L(P)$ such that $ab|c \in \R(\P)$.  The following observation concerning singleton profiles will be useful.

\begin{observation}\label{obs:tripletDisc}
Let $T$ be a phylogenetic tree with $|L(T)| > 2$.  Let $u_1, \dots , u_p$ be the children of $r(T)$.  Then, the connected components of 
$\TG(\{T\})$ are $L(u_1), \dots, L(u_p)$, where $p \ge 2$.
\end{observation}

\section{Testing Compatibility\label{sec:testCompat}}

Here we describe our compatibility algorithm and prove its correctness.  We begin with some definitions.  

Let $U$ be a subset of $V(\P)$ and let $L(U)$ denote $\bigcup_{u \in U} L(u)$. Then, $G_\P(U)$ denotes the graph with vertex set $U$ and where $u, v \in U$ are joined by an edge if and only if $L(u) \cap L(v) \neq \emptyset$. That is, $G_\P(U)$ is the intersection graph of the clusters associated with the nodes in $U$.  For each $i \in [k]$, let $U(i) = U \cap V(T_i)$.  We say that $U$ is \emph{valid} if, for each $i \in [k]$, 
\vspace{-0.1cm}
\begin{enumerate}[(V1)]
\vspace{-0.5\parskip}
 \itemsep1pt \parskip0pt \parsep0pt
\item
if $|U(i)| \ge 2$, then there exists a node $v \in V(T_i)$ such that
$U(i) \subseteq \Ch(v)$ and
\item
$L(U(i)) = L(T_i) \cap L(U)$.
\end{enumerate}

\vspace{-0.1cm}

Observe that the set $\Uinit$ defined as follows is valid.
\begin{equation}\label{eqn:U0}
\Uinit = \{r(T_i) : i \in [k]\}
\end{equation}
Note that $L(\Uinit) = L(\P)$.
From this point forward, we assume that $G_\P(\Uinit)$ is connected. No generality is lost by doing so.  To see why, observe that If $G_\P(\Uinit)$ is not connected, then $\P$ can be partitioned into a collection of species-disjoint profiles $\P_1, \dots, \P_r$ such that $\P$ is compatible if and only if $\P_j$ is compatible for all $j \in [r]$.

The next observation follows from the definition of a valid set.

\begin{observation}\label{lem:restriction}
If $U$ is a valid subset of $V(\P)$, then, for each $i \in [k]$, $\Cl(T_i|L(U)) = \{L(U(i))\} \cup \{L(v) : v \text{ is a descendant of a node in } U(i)\}$. 
\end{observation}

Together with Lemma~\ref{lm:compat}, Observation \ref{lem:restriction} shows that $T_i|L(U)$ is completely determined by the descendants of $U(i)$.

A valid subset $U$ of $V(\P)$ is \emph{compatible} if there exists a phylogenetic tree $T$ with $L(T) =  L(U)$ that displays $T_i|L(U)$ for every $i \in [k]$.  If such a tree $T$ exists, we say that $T$ \emph{displays $U$}. 

\begin{lemma}\label{lm:substructure}
Profile $\P$ is compatible if and only if every valid subset of $V(\P)$ is compatible.
\end{lemma}

\begin{proof}
($\Leftarrow$) If every valid subset of $V(\P)$ is compatible, then, in particular, so is the set $\Uinit$ of Equation \eqref{eqn:U0}. Let $T$ be a tree that displays $\Uinit$.  Then, $L(T) = L(\Uinit) = L(\P)$.  Thus, for every $i \in [k]$, $T_i|L(T) = T_i$, and thus $T$ displays $T_i$.  Hence, $\P$ is compatible.

($\Rightarrow$) Suppose $\P$ is compatible, but there is a valid subset $U$ of $V(\P)$ that is not compatible.  Let $T$ be a tree that displays $\P$. But then $T|U$ displays $U$, a contradiction.
\end{proof}

\BuildG\ (Algorithm~\ref{alg:BuildG}), which is closely related to Semple and Steel's \Build\ algorithm~\cite{SempleSteel03}, determines whether a valid set $U \subseteq V(\P)$ is compatible. 
The key difference between \BuildG\ and \Build\ is that the latter uses the triplet graph $\TG(\P)$, while \BuildG\ uses the graph $G_\P(U)$, for different subsets $U$ of $V(\P)$.  As we show in Lemma \ref{lm:Wi}, the two graphs are closely related.  Nevertheless, $G_\P(U)$ offers some computational advantages over the triplet graph.  Intuitively, this is because $G_\P(U)$ is a more compact representation of the triplets in $\R(\P)$.  

$\BuildG(U)$ attempts to build a tree $T_U$ for $U$.
Step~\ref{alg:rU} initializes the root of $T_U$.
If $L(U)$ consists of one or two species, then $U$ is trivially compatible; Steps~\ref{alg:size1}--\ref{alg:size2Do} handle these cases.  The loop in lines~\ref{alg:forSemiU}--\ref{alg:forSemiULoop} identifies the indices $i \in [k]$ such that $U(i)$ is a singleton. For each such $i$, it removes the single element $v$ in $U(i)$ and replaces $v$ by its children in $T_i$.  As we argue in the proof of Theorem~\ref{thm:BuildG}, when $\P$ is compatible, all such nodes $v$ map to the same node $w$ in the $T$ that displays $\P$, in the sense that $L(w)$ is the smallest cluster in $T$ such that $L(v) \subseteq L(w)$\footnote{Thus, $v$ plays the role of a \emph{semi-universal node}, in the sense of Pe'er et al.~\cite{PeerShamirSharan04}.}.
In Theorem~\ref{thm:BuildG}, we also show that, if $G_\P(U)$ remains connected after steps \ref{alg:forSemiU}--\ref{alg:forSemiULoop}, then $U$ is incompatible.  This case is handled in Line~\ref{alg:incompat}.  Otherwise, Lines \ref{alg:recurseBegin}--\ref{alg:recurseEnd} recursively process each connected component of $G_\P(U)$.  If the recursive calls succeed in finding trees for all the components, these trees are assembled into a phylogeny for $U$ by joining them to the root created in Step~\ref{alg:rU}.  If any of the recursive calls determines that a component is incompatible, then $U$ is declared to be incompatible.

\begin{algorithm}[t]
\SetAlgoLined
\SetNoFillComment
\DontPrintSemicolon
\KwIn{A valid set $U \subseteq V(\P)$.}
\KwOut{A tree $T_U$ that displays $U$, if $U$ is compatible; \texttt{incompatible} otherwise.
}
Create a node $r_U$ \label{alg:rU}\;
\If{$|L(U)| = 1$\label{alg:size1}}
{\Return the tree consisting of node $r_U$, labeled by the single species in $L(U)$\;}
\If{$|L(U)| = 2$\label{alg:size2}}
{\Return the tree consisting of node $r_U$ and two children, each labeled by a different species in $L(U)$ \label{alg:size2Do}\;}
\ForEach{$i \in [k]$ such that $|U(i)| = 1$\label{alg:forSemiU}}
{
Let $v$ be the single element in  $U(i)$\label{alg:forSemiULoopLet}\;
$U = (U \setminus \{v\}) \cup \Ch(v)$\label{alg:forSemiULoop}
}

Let $W_1, W_2, \dots, W_p$ be the connected components of $G_\P(U)$ \label{alg:connected}

\If{$p = 1$}
{
	\Return \incompatible\label{alg:incompat}\;
}
\ForEach{$j \in [p]$ \label{alg:recurseBegin}}
{
	Let $t_j = \BuildG(W_j)$\; \label{alg:recurse}
	\If{$t_j$ is a tree}
	{
		Add $t_j$ to the set of subtrees of $r_U$
	}
	\Else
	{
		\Return \incompatible \label{alg:recurseEnd}
	}
}
\Return the tree with root $r_U$ \;
\caption{\BuildG$(U)$}\label{alg:BuildG}
\end{algorithm}
\vspace{2\parsep}


The correctness of \BuildG\ relies on two lemmas, the first of which can be proved using induction.

\begin{lemma}\label{lm:subtrees}
If, given a valid set $U \subseteq V(\P)$, $\BuildG(U)$ returns a tree $T_U$, then $T_U$ is a phylogenetic tree such that $L(T_U) = L(U)$.
\end{lemma}

The next lemma is central to the correctness proof of  \BuildG.

\begin{lemma}\label{lm:Wi}
Let $W_1, \dots , W_p$ be the connected components of $G_\P(U)$ at step~\ref{alg:connected} of \BuildG$(U)$, for some valid set $U \subseteq V(\P)$.   
Then, 
\vspace{-0.1cm}
\begin{enumerate}[(i)]
\vspace{-0.5\parskip}
 \itemsep1pt \parskip0pt \parsep0pt
\item
for each $j \in [p]$, $W_j$ is a valid set, and
\item
the connected components of $\TG(\P|L(U))$ are precisely $L(W_1), \dots , L(W_p)$.
\end{enumerate}
\end{lemma}

\begin{proof}
(i) Let $\Ubef$ and $\Uaft$ denote the values of $U$ before and after the executing steps  \ref{alg:forSemiU}--\ref{alg:forSemiULoop}.  Each element of $\Uaft$ is either an element of $\Ubef$ or a child of some $v \in \Ubef$.  Indeed, in the latter case, every child of $v$ is in $\Uaft$.  Thus, since, by assumption, $\Ubef$ is valid, and for every non-leaf node $v$, $L(v) = \bigcup_{w \in \Ch(v)} L(w)$, $\Uaft$ must also be valid.  Part (i) follows.  

(ii) 
%
We first show that the following holds after steps \ref{alg:forSemiU}--\ref{alg:forSemiULoop}.
\begin{quote}
\textbf{Claim 1.}
\emph{Let $a$ and $b$ be any two species in $L(U)$. Then, $(a,b)$ is an edge in $\TG(\P|L(U))$ if and only if there exists a node $v \in U$ such that $a,b \in L(v)$.}
\end{quote}
\begin{quote}
\textit{Proof of claim.}
Observe that, after steps \ref{alg:forSemiU}--\ref{alg:forSemiULoop}, $|U(i)| \neq 1$, for each $i \in [k]$.

($\Leftarrow$) Suppose that $(a,b)$ is an edge in $\TG(\P|L(U))$. Then, there is an $i \in [k]$ such that $ab|x \in \R(T_i | L(U))$.  Thus, there must be a proper descendant $w$ of $r(T_i | L(U))$ such that $\{a,b\} \subseteq L(w)$. Observation~\ref{lem:restriction} and the fact that $|U(i)| > 1$ imply that $L(w) \subseteq L(v)$ for some $v \in U(i)$.

($\Rightarrow$) Suppose that there is an $i \in [k]$ such that $a, b \in L(v)$  for some $v \in U(i)$.  Choose a node $v' \in U(i) \setminus \{v\}$ --- such a $v'$ must exist, since $|U(i)| \ge 2$ --- and choose some $x \in L(v')$.  Then, $ab|x \in \R(T_i|L(U))$, and, hence, $(a,b)$ is an edge of $\TG(\P|L(U))$. 
\hfill$\Box$
\end{quote}

Observe that both $\Pi_1 =  \{A : A$ is a connected component of $\TG(\P|L(U))\}$ and $\Pi_2 = \{L(W) : W$ is a connected component of $G_\P(U)\}$ are partitions of $L(U)$. We prove that $\Pi_1 = \Pi_2$ by showing that (a) for each connected component $A$ of $\Gamma(\P|L(U))$ there exists a connected component $W$ of $G_\P(U)$ such that $A \subseteq L(W)$, and (b)
for each connected component $W$ of $G_\P(U)$ there exists a connected component $A$ of $\Gamma(\P|L(U))$ such that $L(W) \subseteq A $.

(a) Let $A$ be any connected component of $\TG(\P|L(U))$.  We argue that any two species $a,b$ in $A$ must be in the same connected component of $G_\P(U)$. 
Let $U_a = \{v \in U: a \in L(v)\}$ and $U_b = \{v \in U : b \in L(v)\}$.  Then, each of $U_a$ and $U_b$ is a clique in $G_\P(U)$.  It thus suffices to show that there is a path between some node in $U_a$ and some node in $U_b$.

By the definition of $A$, there exists a path between $a$ and $b$ in $\TG(\P|L(U))$. Suppose this path is $\rho =\langle a_1, \dots , a_m\rangle$, where $a_1 = a$ and $a_m= b$.  
By Claim 1, for each $l \in [m-1]$, there exists a node $w_l \in U$ such that $\{a_l, a_{l+1}\} \subseteq L(w_i)$.
For each $l \in [m-2]$, $L(w_l) \cap L(w_{l+1}) \neq \emptyset$, so there is a edge between $w_i$ and $w_{i+1}$ in $G_\P(U)$.  Hence, $\pi = \langle w_1, \dots, w_{m-1}\rangle$ is a path from $w_1$ to $w_{m-1}$ in $G_\P(U)$. By the definition of $\rho$, $a \in L(w_1)$ and $b \in L(w_{m-1})$, so $w_1 \in U_a$ and $w_l \in U_b$.  This completes the proof of part (a).

(b) Let $W$ be any connected component of $G_\P(U)$.  If $|L(W)| = 1$, the statement holds trivially, so assume that $|L(W)| > 1$.  We argue that any two species $a,b$ in $L(W)$ are in the same connected component of $\TG(\P|L(U))$. 
Let $v_a$ and $v_b$ be nodes in $W$ such that $a \in L(v_a)$ and $b \in L(v_b)$. If $v_a = v_b$, then, by Claim 1, $(a, b)$ is an edge of $\TG(\P|L(U))$, and we are done. So, suppose instead that $v_a \neq v_b$.  

Let us call a path $\pi$ from $v_a$ to $v_b$ \emph{good} if $|L(w)| > 1$ for every node $w$ in $\pi$. 
We claim that there exists a good path from $v_a$ to $v_b$.  To prove this claim, we first argue that we can choose $v_a$ and $v_b$ such that $|L(v_a)|, |L(v_b)| > 1$.  Indeed, consider the case of species $a$ (the case for $b$ is analogous).  If $|L(v)| = 1$ for every node $v \in W$ such that $a \in L(v)$, then we would have $|L(W)| = 1$, contradicting our assumption that $|L(W)| > 1$.  
Now, suppose the path $\pi$ from $v_a$ to $v_b$ has a node $w \notin \{v_a, v_b\}$ such that $|L(w)| = 1$. Let $w'$ and $w''$ be the predecessor and successor of $w$ in $\pi$.   Then, $L(w') \cap L(w'') = L(w) \neq \emptyset$, so there is an edge between $w'$ and $w''$.  Thus, we can delete $w$ from $\pi$ and the resulting sequence remains a path between $v_a$ and $v_b$.

Let $\pi = \langle w_1, \dots , w_l \rangle$, where $w_1 = v_a$ and $w_l = v_b$, be a good path from $v_a$ to $v_b$ in $G_\P(U)$.  Choose a sequence of species $\rho = \langle c_1, \dots , c_{l+1} \rangle$, where $c_1 = a$, $c_{l+1} = b$ and, for each $j \in [l]$, $c_j, c_{j+1} \in L(w_j)$ and $c_j \neq c_{j+1}$.  Note that such a choice is always possible.  Then, by Claim 1, $(c_j,c_{j+1})$ is an edge of $\TG(\P|L(U))$.  Hence, $\rho$ is a path from $a$ to $b$ in $\TG(\P|L(U))$.
\end{proof}

We are now ready to prove the correctness of \BuildG.

\begin{theorem}\label{thm:BuildG}
Let $\Uinit$ be the set defined in Equation \eqref{eqn:U0}.  Then, \BuildG$(\Uinit)$ either (i) returns a tree $T$ that displays $\P$, if $\P$ is compatible, or (ii) returns \texttt{incompatible} otherwise.
\end{theorem}

\begin{proof}
We first argue that if  \BuildG$(\Uinit)$ outputs \incompatible, $\P$ is indeed incompatible.  Assume, on the contrary, that $\P$ is compatible.  Then, there must be a call \BuildG$(U)$ for some valid subset $U$ such that $|L(U)| > 2$, in which the graph $G(U)$ of step \ref{alg:connected} has a single connected component, $W_1 = U$. By Lemma~\ref{lm:substructure}, $U$ must be compatible, so there exists a phylogeny $T_U$ that displays $U$.  By Observation~\ref{obs:tripletDisc}, $\TG(\{T_U\})$ has at least two connected components $A$ and $B$.    By Lemma~\ref{lm:Wi}(ii), however, $\TG(\P|L(U))$ is connected, so there exist species $a \in A$ and $b \in B$ such that $ab|c \in \R(\P|U)$. But $ab|c \notin \R(T)$, and, by Lemma~\ref{lm:triples}, $T$ does not display some tree in $\P|L(U)$, a contradiction. Thus, $G(U)$ has at least two components. 

Now, suppose that  \BuildG$(\Uinit)$ returns a tree $T$.  We prove that $T$ displays $\P$ by arguing that for each $i \in [k]$ there is a mapping $\phi_i: V(T_i) \rightarrow V(T)$ that maps every node $v \in V(T_i)$ to a node $\phi_i(v) \in V(T)$ such that $L(v) \subseteq L(\phi_i(v))$. 

By Lemma~\ref{lm:subtrees}, each recursive call $\BuildG(U)$ returns a phylogenetic tree $T_U$ for $L(U)$. Let $r_U$ denote the root of $T_U$.  We have two cases.

\textbf{Case (i): $|L(U)| \le 2$.}  For each $i \in [k]$, we must have $|U(i)| \in \{0,1\}$.  We only need to consider the case where $|U(i)| = 1$.  Let $v$ be the single node in $U(i)$.  Note that $L(v) \subseteq  L(r_U)$.  Thus, we make $\phi_i(v) = r_U$.  If $|L(U(i))| = 1$, we are done.  Otherwise, $|L(U(i))| = 2$.   Then, $v$ has two children, $v_1$ and $v_2$, both leaves, labeled with, say, species $s_1$ and $s_2$, respectively.  Node $r_U$ also has two children, $r_1$ and $r_2$.  Assume, without loss of generality, that these children are labeled with species $s_1$ and $s_2$, respectively.  Then, $L(v_j) = L(r_j)$ for $j \in \{1,2\}$.  Therefore, we make $\phi_i(v_j) = r_j$ for each $j \in \{1,2\}$.

\textbf{Case (ii): $|L(U)| > 2$.}  Let $\Ubef$ be the value of $U$ before entering the loop of lines \ref{alg:forSemiU}--\ref{alg:forSemiULoop}, and let $\Uaft$ be the value of $U$ at line~\ref{alg:connected}, after the loop of lines \ref{alg:forSemiU}--\ref{alg:forSemiULoop} terminates. Let $\Urem = \{v \in \Ubef :  v \in \Ubef(i)$ for some $i \in [k]$ such that $|\Ubef(i)| = 1\}$.  Then $\Uaft = (\Ubef \setminus \Urem) \cup \{u \in \Ch(v) : v \in \Urem\}$.   Assume inductively that every descendant of a node in $\Uaft$ is mapped to an appropriate node in $T_U$.  It therefore suffices to establish mappings for the nodes in $\Urem$.  Now, for every $v \in \Urem$, $L(v) \subseteq L(r_U)$.  Thus, we make $\phi(v) = r_U$ for every $v \in \Urem$.
\end{proof}

\section{Implementation\label{sec:implementation}}

We now explain how to implement \BuildG\ in order to solve the tree compatibility problem in $O(\MP \log^2 \MP)$ time.  Consider a call to $\BuildG(U)$.  Recall that we can assume that $G_\P(U)$ is connected.  $\BuildG(U)$ requires the following three pieces of information.
\vspace{-0.1cm}
\begin{enumerate}[({G}1)]
\vspace{-0.5\parskip}
 \itemsep1pt \parskip0pt \parsep0pt
\item\label{item:LU}
\emph{The value of $|L(U)|$.}  This number is needed in Lines \ref{alg:size1} and \ref{alg:size2} of \BuildG.
\item\label{item:singleton}
\emph{The set $J(U)$ of all $i \in [k]$ such that $|U(i)| = 1$.} Set $J(U)$ contains the indices $i$ considered in Lines \ref{alg:forSemiU}--\ref{alg:forSemiULoop} of \BuildG.
\item\label{item:Ui}
\emph{The set $U(i) = U \cap V(T_i)$ for each $i \in [k]$.}  For each $i \in J(U)$, $U(i)$ contains precisely the element $v$ used in Lines \ref{alg:forSemiULoopLet} and \ref{alg:forSemiULoop} of \BuildG.
\end{enumerate}
%
It is straightforward to obtain  (G\ref{item:LU}), (G\ref{item:singleton}), and (G\ref{item:Ui})
for the valid set $\Uinit$ of Equation \eqref{eqn:U0}:  $|L(\Uinit)| = n$,  $J(\Uinit) = [k]$, and, for every $i \in [k]$, $\Uinit(i) = \{r(T_i)\}$.
Now assume that we have (G\ref{item:LU}), (G\ref{item:singleton}), and (G\ref{item:Ui})  at the beginning of some  call to  $\BuildG(U)$. Steps~\ref{alg:forSemiU}--\ref{alg:forSemiULoop} modify $U$ and, therefore, $G_\P(U)$.  Suppose that, after Line  \ref{alg:connected}, $G_\P(U)$ has more than one connected component. We need to compute  (G\ref{item:LU}), (G\ref{item:singleton}), and (G\ref{item:Ui}) for each connected component, in order to pass this information to the recursive calls in Line~\ref{alg:recurse}.  That is, if $p > 1$, for each $j \in [p]$, we need to compute $|L(W_j)|$, $J(W_j)$, and $W_{j}(i) = W_j \cap V(T_i)$, for each $i \in [k]$.

We use the dynamic graph connectivity data structure by Holm et al.\ \cite{HolmLichtenbergThorup:2001}.  We refer to this data structure as \emph{HDT}.  HDT guarantees that, if we start with no edges in a graph with $N$ vertices, the amortized cost of each update is $O(\log^2 N)$.  For efficiency, however, do not use HDT directly on $G_\P(U)$.  The reason is that the edges of $G_\P(U)$ are defined via intersections of sets of species, which could make it costly to determine the new nodes and edges created as a result of Step~\ref{alg:forSemiULoop}.  
We avoid this problem through an indirect approach that uses an auxiliary graph $H_\P(U)$, defined below.   
As we shall see, $H_\P(U)$ offers another advantage over $G_\P(U)$:  maintaining $H_\P(U)$ only requires handling deletions, but maintaining $G_\P(U)$ additionally requires handling insertions.


We define $H_\P(U)$ as a subgraph of the graph $H_P$ constructed as follows.
For each species $s \in L(\P)$, create a new node $x_s \notin V(\P)$, and let $X_\P = \{x_s : s \in L(\P)\}$.
Then, $H_\P$ is the graph whose vertex set is $V(\P) \cup X_\P$ and whose edge set is
$E(\P) \cup \{(u,x_s) : u$ is a leaf in $T_i$,  for some $i \in [k]$, such that $\lambda(u) = s\}$. Note that $H_\P$ has $O(\MP)$ nodes and edges, and can be constructed from $\P$ in $O(\MP)$ time.  $H_\P$ is essentially the  \emph{display graph} for $\P$  \cite{BryantLagergren06}.  The display graph is the result of glueing together leaves in $\P$ labeled by the same species.   Contrast this with $H_\P$, which connects leaves with a common label through nodes in $X_\P$.  This minor difference with respect to the display graph serves to simplify our presentation.

Given a valid subset $U$ of $V(\P)$, we define $H_\P(U)$ as the subgraph of $H_\P$ induced by $\{v : v$ is a descendant of some node $u \in U\} \cup \{x_s  \in X_\P: s \in L(U)\}$. Note that $H_\P(\Uinit) = H_\P$.   See Figure \ref{fig:DisplayGraph}.

\begin{figure}
\begin{center}
\includegraphics[scale=0.34]{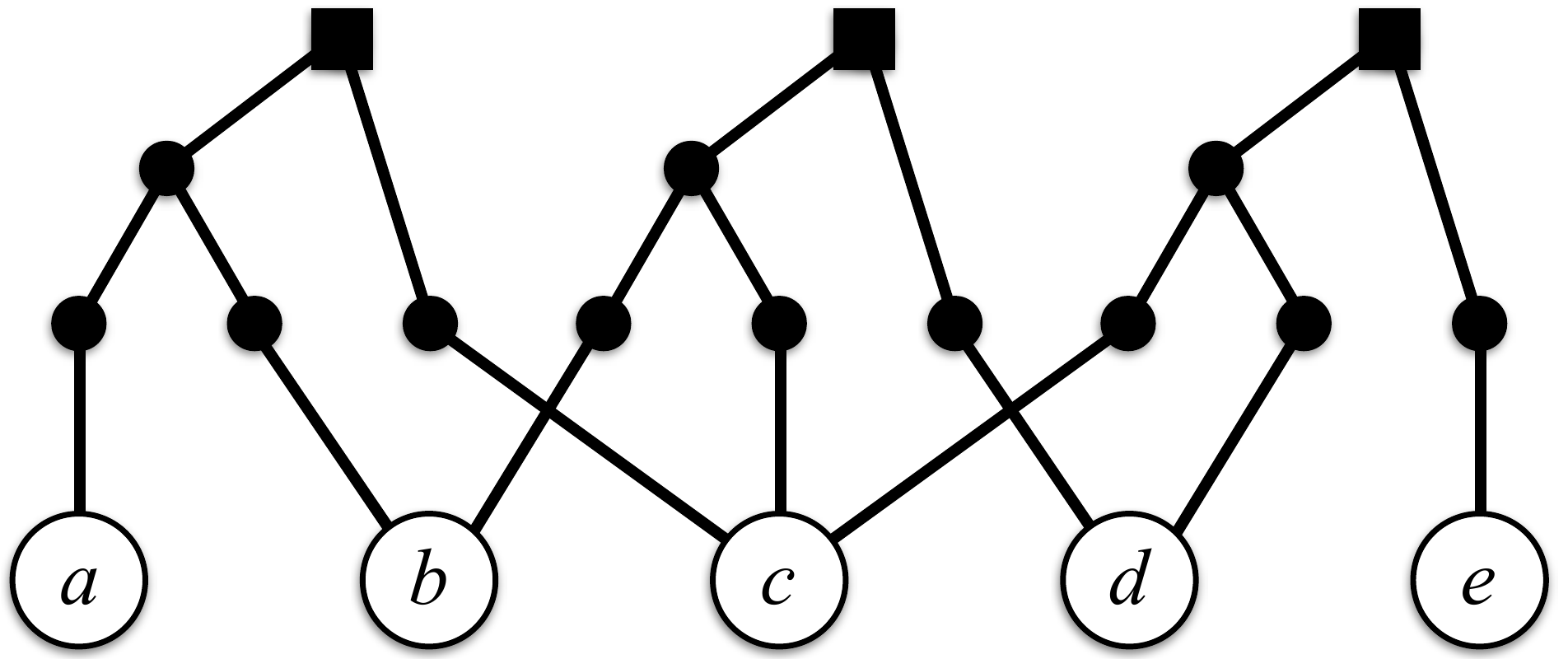}
\end{center}
\vspace{-3\parsep}
\caption{The graph $H_\P(\Uinit)$ for the profile $\P$ of Figure~\ref{fig:Profile}.  The nodes of $\Uinit$ are drawn as squares.
Nodes in the set $\{x_s : s \in L(\P)\}$ are labeled with the corresponding species.  Species labeling the leaves of trees in $\P$ are omitted.}
\label{fig:DisplayGraph}
\end{figure}

The next result states the basic properties of $H_\P(U)$.

\begin{lemma}\label{lm:GnH}
The following statements hold for any valid subset $U$ of $V(\P)$. 
\vspace{-0.1cm}
\begin{enumerate}[(i)]
\vspace{-0.5\parskip}
 \itemsep1pt \parskip0pt \parsep0pt
\item
Let $v$ be a node in $U$. If $U' = (U \setminus \{v\}) \cup \Ch(v)$, then $H_\P(U')$ is obtained from $H_\P(U)$ by deleting $v$ and every edge $(v,u)$ such that $u \in \Ch(v)$.
\item
Any two nodes in $U$ are in the same connected component in $G_\P(U)$ if and only if they are in the same connected component of $H_\P(U)$. 
\end{enumerate}
\end{lemma}

\begin{proof}
The proof of part (i) is trivial, so we focus on part (ii). We argue that, for any two nodes $v, w \in U$, $v$ and $w$ are in the same connected component in $G_\P(U)$ if and only if they are in the same connected component of $H_\P(U)$.

($\Leftarrow$) Suppose $v$ and $w$ are in the same connected component of $H_\P(U)$. Then, there exists a path $\pi$ between $v$ and $w$ in $H_\P(U)$. Let $\pi' = \langle u_1, \dots , u_m\rangle$ be the sequence obtained by striking out from $\pi$ all nodes not in $U$.  Note that $u_1 = v$ and $u_m= w$.  For each $i \in [m-1]$, the subpath of $\pi$ between $u_i$ and $u_{i+1}$ contains at least one node from $X_\P$, say $x_s$, for some $s \in L(\P)$.  But then $s \in L(u_i) \cap L(u_{i+1})$, so there is an edge between $u_i$ and $u_{i+1}$ in $G_\P(U)$. Thus, $\pi'$ is a path between $u_i$ and $u_{i+1}$ in $G_\P(U)$. Hence, $v$ and $w$ are in the same connected component in $G_\P(U)$.

($\Rightarrow$) Suppose $v$ and $w$ are in the same connected component of $G_\P(U)$. Then, there exists  a path between $v$ and $w$ in $G_\P(U)$. Let $\langle u_1, \dots , u_m\rangle$ be one such path, where $u_1 = v$ and $u_m= w$. By definition of $G_\P(U)$, for each $i \in [m-1]$, $L(u_i) \cap L(u_{i+1}) \neq \emptyset$.  Pick any species $s \in L(u_i) \cap L(u_{i+1})$. Hence, in $H_\P(U)$ there is a path from $u_i$ to $u_{i+1}$ that goes through $x_s$. Since this holds for each $i \in [m-1]$, there exists a path from $u_1 = v$ to $u_m = w$ in $H_\P(U)$. Thus, $v$ and $w$ are in the same connected component in $H_\P(U)$.
\end{proof}

By Lemma~\ref{lm:GnH}(ii), the connected components $W_1, \dots , W_p$ of $G_\P(U)$ can be put into a one-to-one correspondence with the connected components $Y_1, \dots , Y_p$ of $H_\P(U)$ so that $W_j = Y_j \cap U$ for each $j \in [p]$. 

We represent $H_\P(U)$ using the aforementioned HDT data structure. For each connected component $Y$ of $H_\P(U)$, we maintain three fields: 
\vspace{-0.1cm}
\begin{enumerate}[(H1)]
\vspace{-0.5\parskip}
 \itemsep1pt \parskip0pt \parsep0pt
\item
$Y.\COUNT$, the cardinality of $Y \cap X_\P$, 
\item $Y.\SEMI$, a doubly-linked list that contains all indices $i \in [k]$ such that $|U(i)| = 1$, and 
\item
$Y.\LIST$, an array where, for each $i \in [k]$, $Y\!.\LIST[i]$ is a doubly-linked list  consisting of the elements of $Y \cap U(i)$.
\end{enumerate}
\vspace{-0.1cm}

Recall that we assume that $G_\P(U)$ is connected at the beginning of a call to $\BuildG(U)$.  Thus, by Lemma \ref{lm:GnH},  $H_\P(U)$ has a single connected component, $Y$.   Then, $|L(U)| = Y.\COUNT$, $J(U) = Y.\SEMI$, and $Y.\LIST[i]$ contains the elements of $U(i)$, for each $i \in [k]$. Thus, the three fields of $Y$ provide $\BuildG(U)$ with the information that it needs  --- that is, (G\ref{item:LU}), (G\ref{item:singleton}), and (G\ref{item:Ui}). 
In particular, they allow us to easily find each node $v$ considered in Line \ref{alg:forSemiULoopLet} of $\BuildG(U)$. Line \ref{alg:forSemiULoop} is then performed as a series of edge deletions, one for each edge $(v,u)$ such that $u \in \Ch(v)$, followed by the deletion of $v$ (we provide further details below).  By Lemma \ref{lm:GnH}(i), this correctly updates $H_\P(U)$.  The deletions break $H_\P(U)$ down into a collection of connected components $Y_1, \dots , Y_p$.  For each $j \in [p]$, $Y_j$ corresponds to a connected component $W_j$ of $G_\P(U)$ that (if $p > 1$) is processed in a recursive call in Line \ref{alg:recurse}.  We need to compute   $Y_j.\COUNT$, $Y_j.\SEMI$, $Y_j.\LIST$ for each $j \in [p]$, in order to provide this information to the recursive calls.   

The total number of edge and node deletions executed by $\BuildG(\Uinit)$ --- including all deletions conducted by the recursive calls --- cannot exceed the total number of edges and nodes in $H_\P$, which is $O(\MP)$.  The HDT data structure allows us to maintain connectivity information throughout the entire algorithm in $O(\MP \log^2 \MP)$.  In the remainder of this section, we show that we can maintain the $\COUNT$, $\SEMI$, and $\LIST$ fields throughout the entire algorithm in total time $O(\MP \log^2 \MP)$.  We also argue that all the required information for $H_\P(\Uinit)$ can be initialized in $O(\MP)$ time.  

Let $\Yinit = V(\P) \cup X_\P$ be the vertex set of $H_\P(\Uinit)$.  Then, $\Yinit$ is the single connected component of $H_\P(\Uinit)$.  We initialize the data fields of $\Yinit$ as follows:
(1) $\Yinit.\COUNT = |L(\P)|$, (2) $\Yinit.\SEMI$ is the set $[k]$, and (3) for each $i \in [k]$, $\Yinit.\LIST[i]$ consists of $r(T_i)$.  Thus, we can initialize all data fields in $O(\MP)$ time.

We assume that every node $v$ in $H_\P(U)$ is either \emph{marked}, if $v \in U$, or \emph{unmarked}, if $v \notin U$.  Initially, each node $v \in \Uinit$ is marked, and every node $v \in \Yinit \setminus \Uinit$ is unmarked.  We also assume that for each node $v$  in $H_\P(U)$, we maintain sufficient information to be able to determine in $O(1)$ time whether $v \in X_\P$ or $v \in V(\P)$, and that, in the latter case, we have $O(1)$-time access to the index $i \in [k]$ such that $v \in V(T_i)$.
For each $i$ such that $Y\!.\LIST[i]$ contains exactly one element, we maintain a pointer from  $Y\!.\LIST[i]$ to the entry for $i$ in $Y.\SEMI$.  This allows us to update $Y.\SEMI$ in $O(1)$ time when $U(i)$ is no longer a singleton. For each marked node $v \in Y$ (so $v \in U$), we maintain a pointer from $v$ to the element in $Y.\LIST[i]$ that contains $v$.  This allows us to update $Y.\LIST[i]$ in $O(1)$ time when $v$ becomes unmarked.


Consider a call to $\BuildG(U)$ for some valid set $U$.   Step  \ref{alg:rU} takes $O(1)$ time.  Since $H_\P(U)$ initially consists of a single connected component, say $Y$, and we have $Y.\COUNT$,  Steps \ref{alg:size1}--\ref{alg:size2Do} also take $O(1)$ time.
Let $H = H_\P(U)$.  
We implement the loop in lines \ref{alg:forSemiU}--\ref{alg:forSemiULoop} as follows. First, we enumerate the indices in $J = J(U)$ in $O(|J|)$ time by listing the elements of $Y.\SEMI$. For each $i \in J$,  we retrieve and remove the single element $v_i$ of  $U(i)$ from $Y.\LIST[i]$, and then delete $i$ from $Y.\SEMI$. This takes $O(1)$ time.  We unmark $v_i$, and for every node $u \in \Ch(v_i)$ we mark $u$ and add it to $Y.\LIST[i]$. This takes $O(1)$ time per edge.  We then successively delete each edge $(v_i,u)$ such that $u \in \Ch(v_i)$, updating (H1)--(H3) for each newly-created component along the way.  Once these edges are deleted, we delete $v_i$ itself.  By Lemma \ref{lm:GnH}(i), the result is the graph $H_\P(U)$ for the new set $U$. Let us focus on how to handle the deletion of a single edge $e = (v_i,u)$. 

Let $Y'$ be the connected component of $H$ that currently contains $v_i$.  We query the HDT data structure to determine, in $O(\log^2 \MP)$ amortized time, whether deleting $(v_i,u)$ splits $Y'$ into two components. If $Y'$ remains connected, no updates are needed.  Otherwise, $Y'$ is split into two parts $Y_1$ and $Y_2$. To fill in the \COUNT, \SEMI, and \LIST\ fields of $Y_1$ and $Y_2$, we use the well-known technique of scanning the smaller component \cite{EvenShiloach:1981}. The HDT data structure maintains the sizes of the various components \cite{HolmLichtenbergThorup:2001}, so we can determine in $O(1)$ which of $Y_1$ and $Y_2$ has fewer nodes.  Suppose without loss of generality that $|Y_1| \le |Y_2|$.  We initialize $Y_2.\COUNT$ and $Y_2.\LIST$ to $Y'.\COUNT$ and $Y'.\LIST$, respectively.  We initialize $Y_1.\COUNT$ to $0$ and $Y_1.\LIST[i]$
to null for each $i \in [k]$. 
We then scan each node $v$ in $Y_1$, and do the following. If $v \in X_\P$, we decrement $Y_2.\COUNT$ and increment $Y_1.\COUNT$.  Otherwise $v \in V(\P)$; assume that $v \in V(T_i)$.   If $v$ is marked, we remove $v$ from $Y_2.\LIST[i]$ and add $v$ to $Y_1.\LIST[i]$. This operation requires at most one update in each of $Y_1.\SEMI$ and $Y_2.\SEMI$; each update takes $O(1)$ time. 

%

We claim that any node $v$ is scanned $O(\log \MP)$ times over the entire execution of $\BuildG(\Uinit)$.  To verify this, let $N(v)$ be the number of nodes in the connected component containing $v$.  Suppose that, initially, $N(v) = N$.  Then, the $r$th time we scan $v$, $N(v) \le N/2^r$.  Thus, $v$ is scanned $O(\log N)$ times.  The claim follows, since $N = O(\MP)$. Therefore, the total number of updates over all nodes is $O(\MP \log \MP)$, and the work per update is $O(1)$. 

To summarize, the work done by $\BuildG$ consists of three parts: (i) initialization, (ii) maintaining connected components, and (iii) maintaining the \COUNT, \SEMI, and \LIST\ fields for each connected component. Part (i) takes $O(\MP)$ time. Part (ii) involves $O(\MP)$ edge and node deletions on the HDT data structure, at an amortized cost of $O(\log^2 \MP)$ per deletion. Part (iii) involves scanning the nodes of our graph every time a deletion creates a new component, for a total of $O(\MP \log \MP)$ scans, at $O(1)$ cost per scan, over the entire execution of $\BuildG$.
This yields our main result.

\begin{theorem}\label{thm:BuildGA}
Let $\Uinit$ be the set defined in Equation \eqref{eqn:U0}.  Then, there exists and implementation of \BuildG\ such that \BuildG$(\Uinit)$ runs in $O(\MP \log^2 \MP)$ time.
\end{theorem}

\section{Discussion\label{sec:discussion}}

A trivial lower bound for the tree compatibility problem is $\Omega(\MP)$, the time to read the input.  Thus, our result leaves us a polylogarithmic factor away from an optimal algorithm for compatibility.  Is it possible to reduce or even eliminate this gap?  The bottleneck is the time to maintain the information associated with the various components of $H_\P(U)$.  It is conceivable that the special structure of this graph and the way the deletions are performed could be used to our advantage.  A second question is how well our algorithm performs in practice.  To investigate this, it should be possible to leverage existing knowledge on the empirical behavior of dynamic connectivity data structures \cite{Iyer+2001}.

\bibliographystyle{abbrv}
\bibliography{amt}

\end{document}